\documentclass[runningheads, envcountsame, a4paper]{llncs}
\usepackage{amsfonts,amssymb,amsmath}
\usepackage[T2A]{fontenc}
\usepackage[utf8]{inputenc}        % Кодировка входного документа;
\usepackage{microtype}%if unwanted, comment out or use option "draft"

\usepackage[hyphens]{url}
\usepackage{verbatim}

\let\leq\leqslant
\let\geq\geqslant

\def\LR{{\mathrm{LR}}}

\def\NN{\mathbb N}

\let\epsilon\varepsilon
\let\eps\varepsilon

\def\abstractname{}

\let\rendmarker\vartriangleleft

\def\A{{\cal A}}

\def\F{{\cal F}}

\usepackage{mathtools}

	\let\lto\leftarrow
	\let\rto\rightarrow
	\let\uto\uparrow
	\let\dto\downarrow

\def\nnext{\mathsf{next}}
\def\prev{\mathsf{prev}}

\let\lendmarker\vartriangleright
\let\rendmarker\vartriangleleft

\title{A Linear-time Simulation of Deterministic $d$-Limited Automata}%{A Linear-time Algorithm for Simulation of Deterministic $d$-Limited Automata}
  
\author{Alexander A. Rubtsov\orcidID{0000-0001-8850-9749}\thanks{Supported by Russian Science Foundation grant 20--11--20203}}%\inst{1}\inst{2}
\institute{
 National Research University Higher School of Economics\\
\email{rubtsov99@gmail.com}
}
\authorrunning{A.~A.~Rubtsov}
\tocauthor{Alexander~A.~Rubtsov}
  \usepackage[makeroom]{cancel} %  для перечёркивания \eps

\usepackage{mathrsfs}% for \mathscr

\usepackage[normalem]{ulem}
\usepackage{color}
\usepackage{xcolor}

\usepackage{enumitem}

\usepackage{marginnote}

\definecolor{WildStrawberry}{HTML}{FF43A4}
\definecolor{DarkGray}{HTML}{666666}

\usepackage{ifthen}

\newcommand{\malert}[2][]{\marginpar{\color{WildStrawberry}\ifthenelse{\equal{#1}{}}{\vspace{-2ex}}{\vspace{-2ex}\vspace{#1}}  #2}}

\newcommand{\mcomment}[2][]{\marginpar{\ifthenelse{\equal{#1}{}}{\vspace{-2ex}}{\vspace{-2ex}\vspace{#1}} \color{DarkGray} #2}}
\newcommand\reduline{\bgroup\markoverwith
      {\textcolor{red}{\rule[-0.5ex]{2pt}{0.4pt}}}\ULon}

\usepackage{dashrule}

\newcommand\alertuline{\bgroup\markoverwith
      {\color{WildStrawberry}{\hdashrule[-0.5ex]{2pt}{0.4pt}{1pt}}}\ULon}

\newcommand\commentuline{\bgroup\markoverwith
			      {\color{DarkGray}{\hdashrule[-0.5ex]{2pt}{0.4pt}{1pt}}}\ULon}

	\usepackage{mathtools} %  для xhookrightarrow   

\DeclareRobustCommand{\dvec}[1]{\ensuremath{\overleftrightarrow{#1}}}
\DeclareRobustCommand{\rvec}[1]{\ensuremath{\overrightarrow{#1}}}
\DeclareRobustCommand{\lvec}[1]{\ensuremath{\overleftarrow{#1}}}

\newcommand{\Gin}{\mathsf{in}}
\newcommand{\Gout}{\mathsf{out}}

\usepackage[linesnumbered]{algorithm2e}

\usepackage{tikz} %Рисование автоматов
\usetikzlibrary{automata, positioning, shapes, snakes, calc, arrows.meta, decorations.markings, matrix}

\usepackage{pgfplots}
\usepackage{pgfmath}

\tikzset{RbtsvGraph/.style={vertex/.style={circle,draw=black,fill=black,thick, inner sep=0pt, minimum size=3.5pt}, style={line width=0.4pt}}} %0.4

\definecolor{xdxdff}{rgb}{0.48018607843137253,0.48018607843137253,1.}
\definecolor{qqqqff}{rgb}{0.,0.,1.}
\definecolor{ffffff}{rgb}{1.,1.,1.}
\definecolor{sqsqsq}{rgb}{0.12548018607843137,0.12548018607843137,0.12548018607843137}
\definecolor{wqwqwq}{rgb}{0.3764705882352841,0.3764705882352841,0.3764705882352841}
\definecolor{cqcqcq}{rgb}{0.7528411764705882,0.7528411764705882,0.7528411764705882}

\def\Gin{\mathsf{in}}
\def\Gout{\mathsf{out}}

\begin{document}

\maketitle
\setcounter{footnote}{0}

\begin{abstract}
A $d$-limited automaton is a Turing machine that may rewrite each input cell at most~$d$ times. 
Hibbard (1967) showed that for every $d \geq 2$ such automata recognize all context-free languages 
and that deterministic $d$-limited automata form a strict hierarchy. 
Later, Pighizzini and Pisoni proved that the second level of this hierarchy coincides 
with deterministic context-free languages (DCFLs). 
We present a linear-time recognition algorithm for deterministic $d$-limited automata in the RAM model, 
thereby extending linear-time recognition beyond DCFLs. 
We further generalize this result to deterministic $d(n)$-limited automata, where the bound $d$ may 
depend on the input length~$n$. 
In addition, we prove an $O(n \cdot k \cdot d(n) + m)$ bound for the membership problem, 
where the input includes both the word and the automaton's description, 
with $m$ denoting the size of the description and $k$ the number of states.
\end{abstract}

\section{Introduction}

Context-free languages (CFLs) play a central role in computer science. 
Their deterministic subclass (DCFLs) is especially important in compiler construction, where parsing is based on the connection between $\LR(1)$-grammars and deterministic pushdown automata (DPDA). 
In 1965 Knuth showed~\cite{Knuth:1965:LR} that $\LR(1)$-grammars generate exactly the class of DCFLs, and DPDAs provide linear-time parsing algorithms for them. 
Thus, DCFLs form a practically significant subclass of CFLs: they are recognizable in linear time, and $\LR(1)$-grammars admit linear-time construction of derivation trees. 
All these linear-time bounds, both in classical parsing theory and in this paper, are measured in the RAM model. 
We emphasize that the situation is different for Turing machines: Hennie showed~\cite{HENNIE1965} that any language recognizable in linear time on a Turing machine is regular, so stronger models such as RAM are required to capture the linear-time parsing of non-regular languages.

It is important to distinguish between two closely related problems. 
We follow the convention that in the \emph{recognition problem}, the language is fixed and the input is only the word~$w$, 
while in the \emph{membership problem}, both a description of the language (for instance, a context-free grammar) and the word~$w$ are given as input. 

According to this convention, the best known upper bound for the \emph{recognition problem} for context-free languages is $O(n^\omega)$, 
where $\omega \leq 2.373$ is the exponent of fast matrix multiplication. 
This bound was obtained by Valiant in 1975~\cite{Valiant75}, and his algorithm decides whether a given word belongs to the language generated by a fixed context-free grammar, 
but it does not construct a parse tree. 
Subsequent work~\cite{Lee02,AbboudFOCS15} confirmed that the same setting --- fixed grammar and variable input word --- is considered, 
and provided strong evidence that this upper bound is hard to substantially improve.

\subsection{$d$-Limited Automata and $d$-DCFLs}

We next recall the notion of $d$-limited automata, introduced by Hibbard~\cite{Hibbard67}. 
A $d$-limited automaton ($d$-LA) is a nondeterministic Turing machine that scans only the cells containing the input word together with end-markers, 
and is allowed to rewrite a symbol in each cell (except the end-markers) only during its first $d$ visits. 
Hibbard showed that for every $d \geq 2$ $d$-LAs recognize precisely the class of context-free languages; it is also known that $1$-LAs recognize exactly the class of regular languages~\cite{WW86}. 
For $d=\infty$, the model coincides with linear-bounded automata, which recognize the class of context-sensitive languages.

For the deterministic case, Pighizzini and Pisoni~\cite{Pig15} proved that deterministic $2$-LAs recognize exactly the deterministic context-free languages (DCFLs) 
by providing an algorithm that transforms a 2-LA into a PDA while preserving determinism; the inverse transformation had already been established by Hibbard~\cite{Hibbard67}. 
Following Hibbard, we call a language recognized by a deterministic $d$-LA a \emph{$d$-deterministic} language ($d$-DCFL) and denote such automata by $d$-DLAs. 
Hibbard also established that the hierarchy is strict: for every fixed $d$, the class of $(d\!+\!1)$-DCFLs properly contains the class of $d$-DCFLs. 

Although DCFLs are widely used, they suffer from certain practical limitations. 
First, they are not closed under reversal: while both 
\[
L_d = \{da^n b^n c^m \mid n,m \geq 0\}, \qquad 
L_e = \{ea^m b^n c^n \mid n,m \geq 0\}
\]
are DCFLs, their union is a DCFL, but the reversal $(L_d \cup L_e)^R$ is not. 
This language can still be recognized by a $3$-DLA, which can scan the input from right to left before simulating a $2$-DLA for $L_d \cup L_e$. 

Second, DCFLs are not closed under union. 
The language 
\[
L_{d,e} = \{a^n b^n c^m \mid n,m \geq 0\} \cup \{a^m b^n c^n \mid n,m \geq 0\}
\]
is a union of two DCFLs but is known to be inherently ambiguous~\cite{ShallitSC}. 
Every DCFL can be generated by an unambiguous grammar, in particular by an $\LR(1)$ grammar, and Hibbard showed~\cite{Hibbard67} that the same holds for $d$-DCFLs. 
Hence the union $\bigcup_{d \geq 1} d$-DCFLs does not cover all CFLs. 
This illustrates that $d$-DCFLs, like DCFLs, are not closed under union; recognizing such languages in linear time via $d$-DLAs requires parallelism.

\subsection{$d(n)$-Limited Automata}

We now consider a natural extension of $d$-LAs. 
Assume that $d$ is not a constant but a function $d(n)$ depending on the input length~$n$. 
The automaton can then rewrite the content of a cell until the number of visits to that cell reaches $d(n)$. 
To our knowledge, this is the first time such a generalization has been considered. 
In classical restrictions on Turing machine computation, the time bound is imposed on the total number of cell visits, 
whereas here we impose a bound on the number of visits per individual cell (after which the cell remains accessible for reading but no longer for rewriting). 

\subsection{Our Contribution}\label{subsect:OurContribution}

We focus on the membership problem for $d(n)$-DLAs. 
Let $m$ be the length of the description of a $d(n)$-DLA~$\A$, 
$k$ the number of its states, 
and $n$ the length of the input word~$w$. 
We present an $O(n \cdot k \cdot d(n) + m)$ algorithm in the RAM model for the membership problem. 
In particular, when $\A$ is fixed and $d(n)=O(1)$ (i.e., for $d$-DLAs), this yields a linear-time algorithm. 
Thus every $d$-DCFL is recognizable in linear time in the RAM model.

Hennie proved in~\cite{HENNIE1965} that every language recognizable in linear time by a deterministic Turing machine is regular 
(and a more general result holds for nondeterministic TMs~\cite{TADAKI2010,PigJALC09}). 
Hence no $d$-DLA with $d \geq 2$ can be simulated by a linear-time Turing machine. 
Guillon and Prigioniero~\cite{Guillon19} showed that every $1$-DLA can be transformed into an equivalent linear-time Turing machine 
(and an analogous result holds for nondeterministic $1$-LAs), 
and a related construction was also used by Kutrib and Wendlandt~\cite{KutribWendlandt:dLA}. 
Their approach relies on Shepherdson’s classical simulation of two-way DFAs by one-way DFAs~\cite{Shepherdson}. 
We build on this idea as well, but it cannot be applied directly because of Hennie’s result: 
otherwise one could simulate a non-regular language in linear time on a Turing machine, contradicting the theorem. 

To overcome this obstacle we transform a classical Turing machine into one that operates on a doubly-linked 
list instead of a tape and adapt Shepherdson’s construction to this model. 
This forms the basis of our membership algorithm. 
We also reinterpret Birget’s algebraic constructions~\cite{Birget89} in graph-theoretic terms, 
which provides subroutines underlying the final version of our algorithm. 
A related algebraic approach was developed by Kunc and Okhotin~\cite{KuncOkhotin2011}, 
who employed transformation semigroups to capture, for each substring, 
the state-to-state behavior of two-way deterministic finite automata. 
This line of work closely mirrors Birget’s method via function composition, 
and our construction follows the same underlying ideas.

We further establish an upper bound of $O(d(n)\cdot n^2)$ steps for direct simulation of $d(n)$-DLAs, 
provided the computation does not enter an infinite loop. 
In particular, this implies an $O(n^2)$ upper bound for $d$-DLAs, which, to the best of our knowledge, was previously open. 
This bound is tight: it is witnessed by the classical $d$-DLA recognizing the language $\{a^n b^n \mid n \geq 0\}$.

From a theoretical perspective, our results show that some CFLs are easy (linear-time recognizable), 
while in general recognition of CFLs may require $O(n^\omega)$ time, 
and by conditional results there must exist hard CFLs (recognizable only in superlinear time). 
We discuss this point in the following subsection.

\subsection{Related Results}

We begin with linear-time recognizable subclasses of context-sensitive languages (CSLs) and context-free languages (CFLs). 
E.~Bertsch and M.-J.~Nederhof~\cite{RegDCFL} showed that 
a nontrivial subclass of CFLs, the regular closure of DCFLs, is recognizable in linear time. 
This class consists of all languages obtained by taking a regular expression and replacing each symbol with a DCFL. 
It evidently contains the aforementioned language $L_{d,e}$ (as a union of DCFLs), so it is a strict extension of DCFLs. 
Note also that $L_{d,e}$ is recognizable by 2-DPDAs.

A broad subclass of CSLs recognizable in linear time is given by the class of languages accepted by two-way 
deterministic pushdown automata (2-DPDAs). 
A linear-time simulation algorithm for 2-DPDAs was obtained by S.~Cook~\cite{Cook1970_2DPDA}. 
This class clearly contains DCFLs, and it also includes the language of palindromes over an alphabet of at least two letters, 
which is a well-known example of a context-free language that is not a DCFL. 

A.~Rubtsov and N.~Chudinov introduced in~\cite{DBLP:conf/mfcs/RubtsovC24,DBLP:journals/corr/abs-2406-14911} 
a computational model, DPPDA, for Parsing Expression Grammars (PEGs). 
This model extends DCFLs, remains recognizable in linear time, and is based on a modification of classical pushdown storage. 
It was also shown that the class of languages recognized by 2-DPPDAs is recognizable in linear time. 
Moreover, they proved that parsing expression languages (the class generated by PEGs) contain a highly nontrivial subclass, 
namely the Boolean closure of the regular closure of DCFLs.

It remains open whether 2-DPDAs, 2-DPPDAs, or PEGs recognize all CFLs. 
However, the works of L.~Lee~\cite{Lee02} and Abboud et al.~\cite{AbboudFOCS15} provide strong evidence 
that this is very unlikely due to complexity-theoretic considerations: 
any CFG parser with time complexity $O(g n^{3-\eps})$, where $g$ is 
the size of the grammar and $n$ the input length, can be efficiently 
converted into an algorithm for multiplying $m \times m$ Boolean matrices in time $O(m^{3-\eps/3})$. 
This naturally raises the question: can 2-DPDAs or 2-DPPDAs simulate $d$-DLAs for $d \geq 3$? 

Finally, T.~Yamakami presented another extension of Hibbard’s 
approach~\cite{yamakami2021behavioralstrengthsweaknessesvarious,DBLP:conf/dcfs/Yamakami25}: 
in several models the input tape is one-way read-only, while the work tape obeys a similar restriction, 
forbidding rewriting beyond the first $d$ visits. 
We leave the application of our simulation technique to such models as a direction for future research.

\section{Definitions}

In this section we give precise definitions of the computational models used in the paper. 
We begin with Hibbard’s original model of $d$-limited automata, introduced in~\cite{Hibbard67}. 
We provide below a concise definition; a more formal equivalent definition can be found in~\cite{Pig15}. 
Next, we introduce a modified variant, called the \emph{deleting automaton}, 
in which the tape is replaced by a doubly linked list and cells may be deleted; 
this auxiliary model is central to our simulation technique. 
Finally, we define $d(n)$-limited automata, a natural extension of 
$d$-limited automata where the rewriting bound depends on the length of the input word. 
For clarity, we formulate our simulation algorithm for the case of a fixed constant~$d$, 
since the extension to $d(n)$ is straightforward.

\subsection{$d$-Limited Automaton}

Let $d \geq 0$ be a fixed integer. 
A deterministic \emph{$d$-limited automaton} ($d$-DLA) is a deterministic single-tape Turing machine 
whose tape initially contains the input word bordered by the left endmarker~$\lendmarker$ 
and the right endmarker~$\rendmarker$. 
Each tape symbol is annotated with an integer in $\{0,\dots,d\}$, called its \emph{rank}. 
Initially, all letters of the input word have rank~$0$, while the endmarkers have rank~$d$. 
Whenever the head visits a cell containing a symbol of rank $r < d$, 
the symbol may be overwritten with a new symbol of rank $r'$ such that $r < r' \leq d$. 
Symbols of rank~$d$ are read-only and cannot be changed. 

Formally, a $d$-DLA $\A$ is defined by a tuple
\[
\A = (Q,\,\Sigma,\,\Gamma,\,\delta,\,q_0,\,F),
\]
where $Q$ is a finite set of states, $\Sigma$ is the input alphabet (symbols of rank~$0$), $\Gamma$ is the tape alphabet with $\Sigma \cup \{\lendmarker, \rendmarker\} \subseteq \Gamma$, 
for each $r$ we let $\Gamma_r \subseteq \Gamma$ denote the set of symbols of rank $r$, 
$q_0 \in Q$ is the initial state, $F \subseteq Q$ is the set of accepting states, 
and $\delta$ is the transition function
\[
\delta : Q \times \Gamma \to Q \times \Gamma \times \{\lto,\rto\},
\]
such that:
\begin{itemize}
    \item for $a_r \in \Gamma_r$ with $r < d$, each transition is of the form 
    \[
    \delta(q, a_r) = (q', a_{r'}, m),
    \]
    where $a_{r'} \in \Gamma_{r'} \setminus \{\lendmarker,\rendmarker\}$, $r < r' \leq d$, and $m\in \{\lto,\rto\}$;
    \item for $a_d \in \Gamma_d$, each transition is of the form 
    \[
    \delta(q,a_d) = (q',a_d,m),
    \]
    where $m \in \{\lto,\rto\}$; moreover, if $a_d = \lendmarker$ then $m = \rto$, 
    and if $a_d = \rendmarker$ then $m = \lto$.    
\end{itemize}

A $d$-DLA starts in state $q_0$ with the head positioned on the first input symbol 
(immediately to the right of the left endmarker $\lendmarker$). 
At each step, given a state $q$ and the symbol $a$ under the head, 
it computes $\delta(q,a) = (q',a',m)$, replaces $a$ with $a'$, updates its state to $q'$, 
and moves the head left or right according to $m$. 
The input is accepted if the automaton reaches a state $q_f \in F$
 while scanning the right endmarker~$\rendmarker$; in this case the computation halts. 

We slightly modify Hibbard’s original definition by requiring the transition function to be total. 
This modification does not change the class of recognizable languages and comes at the cost of a single additional state.

\subsection{Deterministic Linked-List Automaton}

We next introduce an auxiliary model, which we call the 
\emph{Deterministic Linked-List Automaton} (DLLA). 
In this model there is no constraint on the number of visits to a cell. 
The tape is replaced by a doubly linked list, so the automaton may delete any cell between the endmarkers 
(but never the endmarkers themselves). 
Formally, a DLLA has the same components as a $d$-DLA (states, input and tape alphabets, initial and accepting states), 
but with a modified transition function:
\[
\delta : Q \times \Gamma \to 
  \bigl(Q \times (\Gamma \cup \{\perp\}) \times \{\lto,\rto\}\bigr) \cup \{\uto\},
\]
where the special symbol~$\perp$ indicates that the current cell is to be deleted immediately after the head leaves it. 
Once a cell is deleted, the head moves directly from its left neighbor to its right neighbor when moving right, 
and symmetrically when moving left. 
If the transition function returns $\uto$, the computation halts and the input is rejected. 
No additional restrictions are imposed on the transition function.

It is easy to see that DLLAs recognize exactly the class of deterministic context-sensitive languages.
They can simulate deterministic linear-bounded automata (DLBA) directly, and conversely, a DLBA can simulate
a DLLA by marking a deleted cell with a special symbol. 
We employ the doubly linked list representation in order to achieve the claimed upper bounds 
for deterministic $d(n)$-DLAs.

\subsection{$d(n)$-Limited Automaton}

For $d$-DLAs it is convenient to associate with each letter $a_r$ its rank $r$, 
representing the number of visits to the corresponding cell. 
For $d(n)$-DLAs this is no longer possible, since the alphabet and the machine 
description are fixed and cannot grow with the input length. 
Instead, in a $d(n)$-DLA each tape cell maintains its own visit counter, 
so the rank is associated with the cell rather than the symbol. 
Formally, each cell (except the endmarkers $\lendmarker, \rendmarker$) contains a pair $(a,e)$, 
where $a \in \Gamma$ and $e$ is a bit. 
The bit $e$ is initially $0$, and remains $0$ while the number of visits to the cell is less than $d(n)$; 
once the number of visits reaches $d(n)$, $e$ is set to $1$, and from that point on the cell becomes read-only. 

This modification of $d$-DLAs does not affect our simulation algorithm. 
The value of $d(n)$ is fixed for a given input length $n$, and the algorithm never relies on $d$ being a 
constant rather than a precomputed parameter. 
Therefore, it suffices to describe the simulation algorithm for $d$-DLAs; the extension to $d(n)$-DLAs is straightforward.

%%%%%%%%%%%%%%%%%%%%%%%%%%%%%%%%%%%%%%%%%%%%%%%%%%%%%%

\section{Linear-Time Simulation Algorithm}

%%%%%%%%%%%%%%%%%%%%%%%%%%%%%%%%%%%%%%%%%%%%%%%%%%%%%%

In this section we present a linear-time simulation algorithm for deterministic $d$-limited automata (DLAs). 
We begin with the \emph{recognition problem}, where the automaton $\A$ is fixed and only the input word $w$ of length $n$ varies. 
Later we extend the construction to the \emph{membership problem}, where both the automaton and the word are part of the input. 

\medskip
\noindent\textbf{High-level idea.} 
Our approach is inspired by Shepherdson’s classical simulation of two-way deterministic finite automata by one-way DFAs~\cite{Shepherdson}. 
The key observation is that whenever a $d$-DLA $\A$ produces a block of cells all of rank $d$, 
the precise contents of these cells are no longer relevant: 
what matters is only how $\A$ can enter and leave this block. 
We therefore compress each such \emph{maximal block} into a single cell 
containing a compact \emph{mapping} that summarizes the block’s effect on the computation. 
If two adjacent blocks are compressed, their mappings can be merged by composition. 

To implement this idea, we simulate $\A$ by a deterministic linked-list automaton $M$ 
that uses a doubly linked list in place of a tape. 
The machine $M$ performs two types of steps:
\begin{itemize}
  \item \emph{$\A$-moves}, which directly simulate moves of $\A$ on symbols of rank $<d$;
  \item \emph{technical moves}, which occur when $M$ encounters a compressed block. 
        In this case $M$ consults the mapping stored in the corresponding cell to decide how $\A$ would leave the block, 
        and moves accordingly.
\end{itemize}
In this way, long stretches of redundant $d$-rank cells are collapsed into constant-size summaries, 
ensuring that each cell contributes only a bounded number of times to the overall running time.

\medskip
\noindent
In the remainder of this section we first present the simulation algorithm for recognition, 
together with a correctness proof and amortized analysis. 
We then introduce the technical machinery of mappings and composition, 
which allows us to extend the algorithm to the membership problem and to prove the claimed bound.

%%%%%%%%%%%%%%%%%%%%%%%%%%%%%%%%%%%%%%%%%%%%%%%%%%%%%%

\subsection{Preparations}

\paragraph{Directed states.} For convenience, we write $\rvec{p}$ (resp. $\lvec{p}$) to denote a state $p$
entered from the left (resp. right). 
Formally, if $\delta(q, X) = (p, Y, m)$ with $m \in \{\lto,\rto\}$,
we abbreviate it as $\delta(q, X) = (\dvec{p}, Y)$, 
where we write $\dvec{p} = (p,m)$. 
We call such pairs \emph{directed states}, and use this notation 
for both $d$-DLAs and DLLAs. 
We write
\[
\lvec{Q} = Q \times \{\lto\},\quad 
\rvec{Q} = Q \times \{\rto\},\quad 
\dvec{Q} = \lvec{Q} \cup \rvec{Q},\]
\[A_{\uto} = A \cup \{\uto\}\;\;\text{for } A \in \{\lvec{Q},\rvec{Q},\dvec{Q}\}.
\]

\paragraph{Mappings.}
The key idea is to collapse long segments of rank-$d$ cells into a 
single object. 
When $\A$ makes the $d$-th visit to a cell and writes a symbol of 
rank $d$, we replace that cell by a \emph{mapping} $f$. 
Formally, a mapping is a function
\[
f : \dvec{Q} \to \dvec{Q}_{\uto};
\]
it specifies, given the entry state and the \emph{entry direction}, 
the exit state and the \emph{exit direction} 
when $\A$ leaves the block (or $\uto$ if it never does).
For a mapping describing a single rank-$d$ cell, 
the entry direction is irrelevant, 
but for multi-cell segments it matters.

\paragraph{Segment traversal.} 
We denote by $W_\A[i]$ the $i$-th cell of $\A$'s tape and
by \[W_\A[l,r] = W_\A[l]\cdots W_\A[r-1]\] the segment of $\A$'s tape.
When we say that the head \emph{enters the segment}~$W_\A[l,r]$ 
\emph{in the directed state} $\dvec{q}$
we mean that $\dvec{q} = \rvec{q}$ corresponds to the head entering $W_\A[l]$ 
from the left, and $\dvec{q} = \lvec{q}$ corresponds to entering $W_\A[r-1]$ 
from the right. 
Symmetrically, $\dvec{p} = \rvec{p}$ means that the head exits 
through $W_\A[r-1]$ to the right, and $\dvec{p} = \lvec{p}$ 
means it exits through $W_\A[l]$ to the left,
and the head \emph{leaves the segment}~$W_\A[l,r]$
 \emph{in the directed state} $\dvec{p}$.

\paragraph{Segment description mappings.} We say that a mapping $f$ \emph{describes} a segment $W_\A[l,r]$ if
\begin{itemize}
    \item all letters in~$W_\A[l,r]$ have rank $d$ 
    \item $f(\dvec{q}) = \dvec{p}$ if the head enters the 
          segment $W_\A[l,r]$ in a directed state $\dvec{q}$, 
          it leaves the segment in a directed state $\dvec{p}$.
    \item $f(\dvec{q}) = \uto$ if the head enters the 
          segment $W_\A[l,r]$ in a directed state $\dvec{q}$, 
          it never leaves the segment (i.e., the computation loops inside).
\end{itemize}

We denote the set of all possible mappings as
\[\F = \{ f \mid f :  \dvec{Q} \to \dvec{Q}_{\uto}\}\]

\paragraph{Operations with mappings.} Let $f$ and $g$ describe segments $W_\A[L,r]$ and $W_\A[r,R]$ respectively. 
We define the \emph{directed composition} $\diamond$ of mappings
 by setting $h = f \diamond g$ whenever $h$ describes 
 the concatenated segment $W_\A[L,R]$.
Assume now that the head either enters the segment $W_\A[r,R]$
in the directed state $\rvec{q}$, or enters $W_\A[L,r]$
in $\lvec{q}$.  We define the \emph{departure function}
\[D : \F\times\F\times\dvec{Q} \to \dvec{Q}_{\uto}, \]
which, given mappings $f, g$ and an entry state $\dvec{q}$, returns the directed state $\dvec{p}$ such that
 $\A$ leaves the concatenated segment $W_\A[L,R]$ in state $p$ and direction $\dvec{p}$.
 If the head never leaves $W_\A[L,R]$, then $D(f,g,\dvec{q}) = \uto$.
 
Finally, we denote by 
$CF: \Gamma_d \to \F$  a function that
for each symbol $X \in \Gamma_d$, 
returns the mapping $CF(X)$ describing the one-cell segment consisting only of~$X$.

Since $Q$ is finite, the set $\F$ is finite as well. 
Therefore all the functions~$CF$, $D$, and the operation~$\diamond$
are computable in constant time for a fixed automaton $\A$,
because the number of mappings is finite. 
In the general case (varying $\A$) we later show that they are computable in $O(|Q_\A|)$.  
The operation $\diamond$ and the departure function $D$ are well-defined; 
this follows directly from the definitions,
and we will later give a formal justification.

%%%%%%%%%%%%%%%%%%%%%%%%%%%%%%%%%%%%%%%%%%%%%%%%%%%%%% 
 
\subsection{Recognition Problem}

We now present a high-level simulation algorithm for the recognition problem.
Detailed constructions of the subroutines will be given later in 
Subsection~\ref{subsect:Membership}, where the membership problem is analyzed. 
The algorithm below both formally defines the DLLA~$M$ that simulates~$\A$, 
and at the same time serves as the procedure for simulating $M$ on a RAM model.

Simulation Algorithm~\ref{Fig::SimulationAlgorithm} is given in pseudocode;  
its description in natural language is as follows. 
Since the DLLA $M$ deletes cells during its run, we refer to the current 
left and right neighbors of the $i$-th cell as $i.\prev$ and $i.\nnext$, respectively.
When we write $j = i.\prev$ we assume that both $i$ and $j$ refer to the indices 
of $\A$’s tape cells; we do not reenumerate the cells of $M$’s tape.

Thus, we denote the $i$-th cell of $M$’s tape by $W_M[i]$. 
If $W_\A[i]$ contains a letter of rank less than $d-1$, then $W_M[i] = W_\A[i]$,
and $M$ behaves exactly as $\A$ (an $\A$-move). 
When $M$ visits a cell for the $d$-th time,
where $\A$ would write a symbol~$X$ of rank~$d$ (not an end-marker), 
$M$ instead writes to that cell a mapping $g = CF(X)$. 
When $M$ writes a mapping $g$ in a cell $i$ for the first time,
it scans the neighbors $i.\prev$ and $i.\nnext$ and performs the procedure 
we call a \emph{deletion scan}. 

\begin{itemize}
  \item If none of the neighbors contains a mapping, the scan is finished.
  \item If only one of the neighbors $i.\prev$ or $i.\nnext$ contains a mapping 
  (say $f$ or $h$), then $M$ replaces $i$ with $f \diamond g$ (or $g \diamond h$) 
  and deletes the neighboring cell.
  \item If both neighbors contain mappings, then $M$ replaces $i$ with 
  $f \diamond g \diamond h$ and deletes both neighboring cells.
\end{itemize}

After a deletion scan the cell $W_M[i]$ contains the resulting mapping, say $g$, 
while both its neighbors contain letters. 
Hence $g$ describes the segment $W_\A[i.\prev+1,\, i.\nnext]$. 
$M$ then moves the head to the same neighbor ($i.\prev$ or $i.\nnext$) 
where $\A$ would be after leaving the rank-$d$ segment $W_\A[i.\prev+1,\, i.\nnext]$.
If the head of $\A$ had not quit the segment $W_\A[i.\prev+1,\, i.\nnext]$
right after visiting $W_\A[i]$, the cell where $\A$ arrives after it exits 
the segment is determined via the departure function~$D$ during the deletion scan.

We have thus described the cases in which $M$ arrives: 
\begin{itemize}
    \item at a cell of rank $r < d$ (from any neighbor), and
    \item at a cell of rank $d$ (from another rank-$d$ cell).
\end{itemize}

It remains to describe the case when $M$ arrives at a cell containing 
a mapping~$f$ in a directed state~$\dvec{q}$, coming from a cell 
with a letter of rank $r < d$ 
(lines~\ref{alg::Fmove:begin}--\ref{alg::Fmove:end}). 
In this case $M$ computes $f(\dvec{q}) = \dvec{p}$ and moves the head to 
the left or right neighbor according to the direction of $\dvec{p}$, 
arriving at that neighbor in state~$p$.

\begin{figure}[h]

\SetKw{Or}{or}
\SetKw{And}{and}

\begin{algorithm}[H]%algorithm2e
%	\tcc{For $d = 0 $ we assume that all letters in the beginning have rank $-1$. See the description in Subsection~\ref{subsect::Algorithm} }
	$\dvec{q} := \lvec{q_0}$;\quad $i := 1$\;		
	\SetKwInOut{Input}{Input}\SetKwInOut{Output}{Output}
	\SetKwProg{Fn}{Function}{ :}{end}	
		\SetKwFunction{KwPrint}{print}

	\While{ no result returned }{	
	  %\Switch{$W_M[i]$}{		    
		\uCase(\tcc*[f]{$\A$-move}){ $W_M[i] \in \Gamma_r \cup \{\lendmarker, \rendmarker\}, r < d-1$}{\label{alg::Amove:begin}	
		%\tcc*[r]{$\A$-move}	
			$(\dvec{p}, a_{r^\prime}) := \delta_\A(q, W_M[i])$;\,\,
			$W_M[i] := a_{r^\prime}$\;\label{alg::Amove:end}
		}
		\uCase(\tcc*[f]{Deletion scan}){$W_M[i] \in \Gamma_{d-1}$}{		
		 $(\dvec{p}, X) := \delta_\A(q, W_M[i])$\;
		 $g := CF(X)$ \tcc*{ $\langle CF \rangle(\A) $ }
		 \If{ $i.\prev > 0$ \And $W_M[i.\prev] \in \F $}{
		 	$f := W_M[i.\prev]$\;
			\lIf{$\dvec{p} = \lvec{p}$}{$\dvec{p} := D(f,g,\lvec{p})$}
			\lIf{ $\dvec{p} = \,\uto$ }{ \KwRet Reject}
			$ g := f\diamond g  $;\,\,
			$i.\prev := (i.\prev).\prev$\tcc*{ $\langle \diamond \rangle(\A)$.}
			  
		 }
		 \If{ $i.\nnext < n+1$ \And $W_M[i.\nnext] \in \F $}{
		 	$h := W_M[i.\nnext]$\;
			\lIf(\tcc*[f]{ $\langle D \rangle(\A)$.}){$\dvec{p} = \rvec{p}$}{ 
				$\dvec{p} = D(g, h,\rvec{p})$}
			\lIf{ $\dvec{p} = \,\uto$ }{ \KwRet Reject}
			$ g :=  g \diamond h  $;\,\,
			$i.\nnext := (i.\nnext).\nnext$\tcc*{ $\langle \diamond \rangle(\A)$.}			  
		 }
		 $W_M[i] := g$\;
		}\Case{$W_M[i] \in \F$}{ \label{alg::Fmove:begin} 
			$f := W_M[i]$\;	%\tcc*[r]{$W_M[i] \in \F$}		
			\leIf{ $f(\dvec{q}) = \,\uto$ }{ \KwRet Reject}
			{$\dvec{p} = f(\dvec{q})$}
		}\label{alg::Fmove:end}
		%}%EoFSwitch
		\leIf{ $\dvec{p} = \lvec{p}$ }{ $i := i.\prev$ }{$i := i.\nnext$}
		$\dvec{q} := \dvec{p}$\;
		\lIf{ $q \in F$ \And $W_M[i] = \,\rendmarker$ }{ \KwRet Accept} 	  
	}	
    \caption{Simulation Algorithm}\label{Fig::SimulationAlgorithm}
\end{algorithm}

\end{figure}

Before proving correctness we fix notation for time indices. 
Let $W_\A^t[i]$ denote the content of cell $i$ of $\A$'s tape after $t$ steps 
of $\A$, and let $W_M^{t'}[i]$ denote the content of cell $i$ of $M$'s tape 
after $t'$ steps of $M$. 
We consider only \emph{regular steps}, i.e. steps performed on symbols of rank $< d$ 
or on endmarkers. 
Every regular step $t$ of $\A$ has a corresponding regular step $t'$ of $M$, 
and the mapping $t \mapsto t'$ is strictly increasing 
(if $t_1 < t_2$ then $t'_1 < t'_2$). 
This correspondence will be used below.

\begin{lemma}\label{Th::Correctness}
For each $d$-DLA $\A$, the corresponding DLLA $M$ simulates $\A$. 
More precisely, there exists an order-preserving correspondence $t \mapsto t'$ such that:
\begin{itemize}
  \item[(i)] for every regular step $t$ of $\A$ with the corresponding step $t'$ of $M$, 
  if $\A$ visits cell $i$ with a symbol of rank $< d$ or an endmarker, 
  then $W_\A^t[i] = W_M^{t'}[i]$;
  \item[(ii)] at steps $t$ and $t'$ $\A$ and $M$ 
              are in the same state when arriving at cell $i$; 
  \item[(iii)] $M$ accepts an input iff $\A$ does.
\end{itemize}
\end{lemma}

\begin{proof}
Let $\A$ perform $N$ moves
\begin{equation}\label{eq::run}
\delta_1, \ldots, \delta_N,\quad 
\delta_i \in Q\times\Gamma\times\{\lto, \rto\},\quad 
\delta_1 = \delta(q_0, W^0_\A[1]) 
\end{equation}
on a fixed input, and either accepts the input or enters a loop.  

We call a move $\delta_i = (q,a,m)$ a \emph{$d$-move} if $a$ has rank $d$
but is not an endmarker; otherwise we call it a \emph{regular move}.  
Thus a run is a sequence~\eqref{eq::run} partitioned into alternating 
segments of regular moves and $d$-moves.

For $M$ we define runs analogously: regular moves are the same, 
while a $d$-move is either a step into a cell containing a mapping, 
or a step of the deletion scan initiated on a cell with a symbol of rank $d-1$.  

We claim:
\begin{enumerate}
  \item[(i)] if we delete all $d$-moves from the runs of $\A$ and $M$, the resulting sequences of regular moves are identical;
  \item[(ii)] after each maximal block of $d$-moves, both automata end in the same cell and in the same state.
\end{enumerate}
From these two properties it follows that $M$ accepts exactly
the same words as $\A$, because accepting configurations are reached 
by regular moves only. The correspondence $t \mapsto t'$ 
between regular steps of $\A$ and $M$ is then precisely 
the index matching described before the lemma.

Note that property (i) follows immediately from (ii), 
since between two blocks of $d$-moves both automata perform the same 
sequence of regular moves. Indeed, once the heads are in the same cell 
with the same symbol of rank $< d$ and in the same state, 
the subsequent regular moves of $M$ coincide with those of $\A$ 
by construction.

It remains to prove (ii). 
Assume that before some $d$-block both $\A$ and $M$ are in the same state 
on the same cell. We distinguish two cases for the first move of 
the $d$-block of $\A$:

\smallskip
\noindent
\textbf{Case 1:} the head visits a cell $i$ containing a symbol $a$ of rank $d-1$.  
Then $\A$ rewrites $a$ to a symbol $X$ of rank $d$.  
In the corresponding move $M$ writes into $i$ the mapping $f_X$ describing 
the one-cell segment at $i$.  
If the $d$-block of $\A$ ends immediately 
(or on the next step $\A$ visits another rank-$(d-1)$ cell), 
then by the definition of $f_X$ both automata end in the same cell and state.  
Otherwise, $\A$ proceeds into a neighboring rank-$d$ cell and eventually 
leaves the contiguous rank-$d$ segment.  
By construction of $M$, after the deletion scan and subsequent 
use of the departure function $D$, $M$ arrives at the same cell and in 
the same state as $\A$.  
Thus both automata synchronize at the end of the $d$-block.

\smallskip
\noindent
\textbf{Case 2:} the head of $\A$ enters a cell of rank $d$.  
Hence $\A$ moves inside a segment of contiguous rank-$d$ symbols, while $M$ is positioned at 
a cell containing a mapping $f$ describing this segment (the invariant maintained by the deletion scan).  
Since $f$ faithfully describes the segment, after $M$ executes the step $f(\dvec{q})$, 
it reaches exactly the same exit cell and state as $\A$.  
If $\A$ then continues into a rank-$(d-1)$ cell, we return to Case~1; otherwise the $d$-block ends with synchronization.

\smallskip
Finally, if in either case $\A$ enters a loop, the corresponding mapping for $M$ returns $\uto$, and $M$ rejects the input.  
Thus (ii) holds, which completes the proof.
\qed
\end{proof}
%%%%%%%%%%%%%%%%%%%%%%%%%%%%%%%%%%%%%%%%%%%%%%%%%

Now we prove that the simulation algorithm for $d$-DLAs works in linear time. 
We present the proof for the general case of a $d(n)$-DLA. 
The simulation algorithm for $d(n)$-DLAs is identical to that for a fixed $d$, 
since its behavior depends only on whether the number of visits to a cell 
is equal to $d(n)-1$ or smaller. 
The counters for cell visits required in the case of $d(n)$-DLAs 
can be implemented in the RAM model with $O(1)$ overhead per operation, 
and thus do not affect the asymptotic running time.

We denote by $\langle F \rangle(\A)$ the time complexity of the operation $F$ 
on a $d$-DLA~$\A$. 
Since we will prove later that the complexity of each auxiliary step 
depending on $\A$ is $O(|Q|)$, we replace 
$\langle CF \rangle(\A)$, $\langle D \rangle(\A)$, and $\langle \diamond \rangle(\A)$ by
\begin{equation}\label{eq:UB}    
\langle U\!B \rangle(\A) \;=\; \langle CF \rangle(\A) \;+\; \langle D \rangle(\A)
\;+\;\langle \diamond \rangle(\A).
\end{equation}

\begin{lemma}\label{th:AmortizedAnalysis}
The automaton $M$ performs $O(d(n)\cdot\langle U\!B\rangle(\A)\cdot n)$ steps 
on processing an input of length~$n$. 
\end{lemma}

\begin{proof}
We use amortized analysis~\cite{CLRS2022}, namely the accounting method. 
Each cell $i \in \{1,\ldots, n\}$ on $M$'s tape has its own budget $B[i]$ 
(credit, in the terminology of~\cite{CLRS2022}). 
We denote by $B^t[i]$ the value of the budget of cell $i$ after step~$t$. 

The budgets are updated according to the following rules: 
\begin{itemize}
  \item $B^{0}[i] = 2d(n)$ for all $i$;
  \item $B^{t}[i] = B^{t-1}[i]-1$ if at step $t$ the head visits cell~$i$ 
        and this cell still contains a letter (i.e., it has been visited fewer than $d(n)$ times);
  \item $B^{t}[j] = B^{t-1}[j]-1$ if at step $t$ the head enters cell~$i$ 
        from cell~$j$, where $W_M[i]$ is a mapping and $W_M[j]$ is a letter;
  \item $B^{t}[i] = B^{t-1}[i]$ otherwise.
\end{itemize}

Budgets are not changed during deletion scans.

\smallskip
Fix a step $t$. 
Suppose that the $i$-th cell currently contains a mapping $f = W^t_M[i]$ 
describing a segment $W_\A[l,r]$. 
Its neighbors $W_M[l-1] = W_\A[l-1]$ and $W_M[r+1]=W_\A[r+1]$ have rank $< d(n)$, 
so each has been visited fewer than $d(n)$ times. 
When the head moves from a neighbor into the segment, that neighbor pays \$1. 
Until the neighbor itself turns into a mapping, it pays at most $d(n)$ times 
for such visits. Once it becomes a mapping, further payments are taken over 
by the new neighbor. 
Thus each cell pays \$1 for each of its own visits and at most \$1 for each 
visit into an adjacent mapping before it itself turns into a mapping. 
Since this can happen only after at most $d(n)$ visits of the cell, 
each cell pays at most $2d(n)$ dollars in total. 
Therefore, the described budget strategy guarantees $B^t[i] \ge 0$ 
for all $t$ and~$i$.

\smallskip
Deletion scans were not counted above. 
Clearly there are at most $O(n)$ scans, since each cell can initiate at most one. 
During one scan, at most two directed compositions are computed, 
each in $O(\langle U\!B\rangle(\A))$. 
Hence all deletion scans together cost $O(n\cdot \langle U\!B\rangle(\A))$ time. 

\smallskip
Summing up, $M$ performs the following kinds of operations:
\begin{enumerate}
  \item moves that end on a letter, but not an endmarker;
  \item moves that end on a mapping;
  \item deletion scans;
  \item moves that end on an endmarker.
\end{enumerate}

By amortized analysis, Cases 1 and 2 together take 
$O(d(n)\cdot \langle U\!B\rangle(\A) \cdot n)$: 
the total number of such moves is $O(d(n)\cdot n)$ 
(since $\sum_{i=1}^n B^0[i] = 2n\cdot d(n)$), 
and each move costs $O(\langle U\!B\rangle(\A))$. 
Case~3 costs $O(n\cdot \langle U\!B\rangle(\A))$, as discussed. 
For Case~4, note that after the head leaves the left endmarker $\lendmarker$ 
on the very first move, each endmarker can only be visited 
when arriving from an inner cell. 
Hence the total number of endmarker visits does not exceed 
the number of visits to all other cells, which is $O(n\cdot d(n))$. 
Since each simulation step of Algorithm~\ref{Fig::SimulationAlgorithm} 
takes at most $O(\langle U\!B\rangle(\A))$ time, 
endmarker visits also cost $O(d(n)\cdot \langle U\!B\rangle(\A) \cdot n)$. 

Thus the total running time of $M$ on an input of length~$n$ 
is $O(d(n)\cdot \langle U\!B\rangle(\A)\cdot n)$.
\qed
\end{proof}

%%%%%%%%%%%%%%%%%%%%%%%%%%%%%%%%%%%%%%%%

\subsection{Membership Problem}\label{subsect:Membership}

To analyze the membership problem for $d$-DLAs we need to formalize 
the auxiliary operations on mappings. 
Recall that mappings represent contiguous segments of rank-$d$ cells 
and that the simulation algorithm relies on three basic subroutines:
\begin{itemize}
  \item the \emph{cell description} function $CF$, 
        which produces the mapping for a single rank-$d$ cell;
  \item the \emph{directed composition} $\diamond$, 
        which merges mappings of adjacent segments into one;
  \item the \emph{departure function} $D$, 
        which determines the exit state and direction 
        when the head is located at the boundary between two adjacent segments,
        i.e., when it enters one segment from the other.
\end{itemize}
We prove in this subsection that all these operations are well-defined and computable 
in $O(|Q_\A|)$ time, so $\langle U\!B\rangle(\A) = O(|Q_\A|)$.

\smallskip
Our constructions rely on graph representations of mappings. 
A mapping $f \in \F$ describing a segment $L$ is encoded by a four-partite 
graph $G_f$ with parts 
$\rvec{L}_{\Gin}, \lvec{L}_{\Gin}, \rvec{L}_{\Gout}, \lvec{L}_{\Gout}$, 
each a copy of $Q_\A$. 
These four parts form a partition of the set $\dvec{Q}$ according to their labeling.
We also adjoin a distinguished sink vertex $(\uto)$ of out-degree~$0$. 
For every $\dvec{q}\in\dvec{L}_{\Gin}$ we add:
\begin{itemize}
    \item an edge $\dvec{q}\to\dvec{p}$ to some $\dvec{p}\in\dvec{L}_{\Gout}$ if $f(\dvec{q})=\dvec{p}$,
    \item or an edge $\dvec{q}\to (\uto)$ if $f(\dvec{q})=\uto$.
\end{itemize}
Thus every vertex has out-degree at most~$1$: 
inputs have either one outgoing edge to an output or to $(\uto)$, 
while outputs and $(\uto)$ have out-degree~$0$. 
We say that $G_f$ \emph{represents} the mapping $f$ (or the segment $L$).

If $g$ describes an adjacent segment $R$, to compute the composition 
$f\diamond g$ and the departure function $D$ we use the 
\emph{intermediate graph}~${\cal I}(f,g)$ obtained 
by gluing the graphs $G_f$ and $G_g$ (with parts labeled by $R$).
Formally, ${\cal I}(f,g)$ is defined as follows: 
part $\rvec{L}_{\Gout}$ is glued with $\rvec{R}_{\Gin}$, 
part $\lvec{L}_{\Gin}$ with $\lvec{R}_{\Gout}$, 
and the two sinks~$(\uto)$ are glued together. 
By gluing two vertices we mean that one of them is deleted and 
all its edges are reattached to the other.
When we glue parts, we glue the vertices carrying the same state label 
(but with opposite directions). 
An illustration of ${\cal I}(f,g)$ is given in Fig.~\ref{Fig:GraphComp}.

\begin{figure}
\begin{center}
\begin{tikzpicture}[
  line cap=round,line join=round,>=Stealth,
  x=1cm,y=1cm,scale=1,
  vertex/.style={circle,fill=black,inner sep=1.6pt},
  Lset/.style={line width=0.9pt,draw=black!60,fill=black!10},
  Rset/.style={line width=0.9pt,draw=black!60,fill=black!1},
  fedge/.style={->,line width=0.9pt,draw=black},                % edges of f (solid)
  gedge/.style={->,line width=0.75pt,draw=black!50},      % edges of g (dashed)
  glue/.style={line width=0.8pt,draw=black!60,dash pattern=on 3pt off 2pt}
]

\clip(1.5,1.5) rectangle (11.5,9.6);

%--- L parts (left segment) ---
\draw[Lset] (2.7,8) ellipse (0.6cm and 1.5cm);   \node at (1.8,8) {$\rvec{L}_{\Gin}$};
\draw[Lset] (6.7,8) ellipse (0.6cm and 1.5cm);   \node at (8.2,8) {$\lvec{L}_{\Gin} = \rvec{R}_{\Gout}$};
\draw[Lset] (3.7,4.5) ellipse (0.6cm and 1.5cm); \node at (2.7,4.5) {$\lvec{L}_{\Gout}$};
\draw[Lset] (5.2,4.5) ellipse (0.6cm and 1.5cm); \node at (6.7,4.5) {$\rvec{L}_{\Gout} = \lvec{R}_{\Gin}$};

%--- R parts (right segment) ---
\draw[Rset] (10.0,8) ellipse (0.6cm and 1.5cm);  \node at (11.0,8) {$\lvec{R}_{\Gin}$};
\draw[Rset] (8.5,4.5) ellipse (0.6cm and 1.5cm); \node at (9.5,4.5) {$\rvec{R}_{\Gout}$};

%--- glue (placed BELOW the ovals) ---
%\draw[glue] (6.6,6.4) -- (8,6.4);
%\node[fill=white,inner sep=1pt] at (7.4,6.1) {$\lvec{L}_{\Gin}\!=\!\lvec{R}_{\Gout}$};

%\draw[glue] (4.6,2.9) -- (5.8,2.9);
%\node[fill=white,inner sep=1pt] at (5.2,2.6) {$\rvec{L}_{\Gout}\!=\!\rvec{R}_{\Gin}$};

%--- vertices (few per part just for illustration) ---
\node[vertex] (Lrin1) at (2.7,9) {};
\node[vertex] (Lrin2) at (2.7,8.25) {};
\node[vertex] (Lrin3) at (2.7,7.5)  {};

\node[vertex] (Llin1) at (6.7,9) {};
\node[vertex] (Llin2) at (6.7,8.25) {};
\node[vertex] (Llin3) at (6.7,7.5)  {};

\node[vertex] (Llout1) at (3.7,5.5) {};
\node[vertex] (Llout2) at (3.7,4.75) {};
\node[vertex] (Llout3) at (3.7,4.0)  {};

\node[vertex] (Lrout1) at (5.2,5.5) {};
\node[vertex] (Lrout2) at (5.2,4.75) {};
\node[vertex] (Lrout3) at (5.2,4.0)  {};

\node[vertex] (Rlin1) at (10.0,9) {};
\node[vertex] (Rlin2) at (10.0,8.25) {};
\node[vertex] (Rlin3) at (10.0,7.5)  {};

\node[vertex] (Rrout1) at (8.5,5.5) {};
\node[vertex] (Rrout2) at (8.5,4.75) {};
\node[vertex] (Rrout3) at (8.5,4.0)  {};

%--- edges of f (solid, black) ---
\draw[fedge] (Lrin1) -- (Lrout1);
\draw[fedge] (Lrin2) -- (Lrout2);
\draw[fedge] (Llin1) -- (Llout3);
\draw[fedge] (Lrin3) -- (Llout1);

%--- edges of g (dashed, gray) ---
\draw[gedge] (Rlin3) -- (Rrout2);
\draw[gedge] (Rlin1) -- (Llin1);
\draw[gedge] (Llin2) -- (Lrout1);
\draw[gedge] (Lrout1) .. controls (6.5,6.2) and (7.2,7) .. (Llin2);
\draw[gedge] (Lrout2) .. controls (6.5,3) and (7.2,3) .. (Rrout3);

%--- uto vertex ---
\node[vertex,label=above:$(\uto)$] (uto) at (8.5, 6.5) {};
\draw[gedge] (Rlin2) .. controls (9,6.8) .. (uto);
\draw[fedge] (Llin3) -- (uto);

%--- legend ---
\node[draw,inner sep=2pt,anchor=west] at (3,2) {
  \begin{tikzpicture}[baseline=-0.5ex]
    \draw[fedge] (0,0)--(0.9,0); \node at (1.4,0) {\scriptsize edges of $f$};
    \draw[gedge] (3.2,0)--(4.1,0); \node at (4.6,0) {\scriptsize edges of $g$};
  \end{tikzpicture}
};

\end{tikzpicture}
\vspace*{-8pt}
\caption{Graph-based computation of $f\diamond g$ and $D$.}
\label{Fig:GraphComp}
\end{center}
\end{figure}

\begin{proposition}\label{Prop:GraphProps}
        The directed composition $h=f\diamond g$ is determined via the intermediate graph
        as follows: $h(\dvec{q})=\dvec{p}$ iff there is a path 
        $\dvec{q}\leadsto\dvec{p}$ to an output vertex, and $h(\dvec{q})=\uto$ iff the
        unique path from $\dvec{q}$ reaches $(\uto)$ or falls into a directed cycle.
        Associativity follows from associativity of concatenating segments (hence of
        gluing graphs).
\end{proposition}

\begin{lemma}\label{lem:UB}
For every $d$-DLA $\A$, the operations $\diamond$, $D$, and $CF$ 
are well-defined and computable in $O(|Q_\A|)$ time, 
assuming that the description of $\A$ is already stored in RAM.
\end{lemma}

\begin{proof}
The function $CF$ is trivially computable in $O(|Q_\A|)$ time: 
for a rank-$d$ symbol $a_d$, the mapping $f=CF(a_d)$ satisfies 
$f(\dvec{q})=\dvec{p}$ iff $\delta_\A(q,a_d)=(\dvec{p},a_d)$.

\smallskip
\SetKwFunction{FindPath}{FindPath}
We next analyze the algorithms for $\diamond$ and $D$. 
Let $V$ and $E$ be the vertex and edge sets of the intermediate graph~${\cal I}(f,g)$
(Fig.~\ref{Fig:GraphComp}). Since each vertex has out-degree at most~1, 
we have $|E|=O(|V|)$ and $|V|=O(|Q_\A|)$. 
We maintain two global arrays: $h$ for the resulting mapping 
and $m$ for marks on vertices. Initially all values of $m$ are set to a placeholder~$\dto$.
These arrays serve as memoization.

The core procedure is the function \FindPath (Algorithm~\ref{Fig:Alg:FindPath}), 
which returns the endpoint of the unique path starting from $\dvec{q}$, 
or~$\uto$ if the path falls into a loop. 
We explain the steps of the algorithm while proving correctness.

% --- Shared globals for the intermediate graph G ---
% V = vertices of G; inputs: \rvec{L}_{\Gin}\cup\lvec{R}_{\Gin};
% outputs: \lvec{L}_{\Gout}\cup\rvec{R}_{\Gout}; single sink: (\uto).
% Arrays: m[\cdot] (marks), h[\cdot] (values for inputs).

\begin{figure}[ht]
\begin{algorithm}[H]
  \SetKwInOut{Input}{Input}\SetKwInOut{Output}{Output}  
  \SetKwProg{Fn}{Function}{ :}{end}

  %%% Subroutine %%%
  \Fn{\FindPath{$\dvec{q}$}}{ \label{alg:FindPath}
    $u := \dvec{q}.\nnext$\;
    \While{$u \not\in \lvec{L}_{\Gout}\cup\rvec{R}_{\Gout}\cup\{(\uto)\}$}{
      \uIf{$m[u] = \dto$}{
        $m[u] := \dvec{q}$;\quad $u := u.\nnext$\;
      }\uElseIf{$m[u] = \dvec{q}$}{
        \KwRet $\uto$\tcc*{Detected a cycle reachable from $\dvec{q}$}
      }\Else{
        \KwRet $h[m[u]]$\tcc*{Path merges with a processed input}
      }
    }
    \KwRet $u$\tcc*{Reached an output or the sink $(\uto)$}
  }

  %%% Main: compute h = f \diamond g on the intermediate graph %%%
  \caption{\textsc{FindPath} on the intermediate graph}\label{Fig:Alg:FindPath}
\end{algorithm}
\end{figure}

\smallskip
\emph{Correctness and complexity of \FindPath.}
Let $u_0=\dvec{q}$ and define $u_{k+1}=u_k.\nnext$. 
The procedure terminates when one of the following holds:
\begin{itemize}
  \item $u_k$ is an output: then \FindPath returns $u_k$;
  \item $u_k=(\uto)$: then \FindPath returns $\uto$;
  \item $u_k$ is marked: let $\dvec{s}=m[u_k]$.
\end{itemize}
If $\dvec{s}=\dvec{q}$, a loop is detected and \FindPath returns $\uto$.
Otherwise \FindPath returns $h(\dvec{s})$, 
since $h(\dvec{s})$ has already been computed and stored in $h[m[u]]$. 
As the paths from $\dvec{q}$ and $\dvec{s}$ merge at the vertex $u_k$,  
their endpoints coincide.

The invariant “if $m[u]=\dvec{s}$ then $h[m[u]]=h(\dvec{s})$” 
is established by Algorithm~\ref{Alg:d-comp-refactored}, 
which computes $h=f\diamond g$ using \FindPath. 

\begin{figure}[ht]
\begin{algorithm}[H]
  \SetKwInOut{Input}{Input}\SetKwInOut{Output}{Output}

  $m[u] := \dto$ for all $u\in V$\tcc*{clear marks}
  $h[\dvec{q}] := \dto$ for all $\dvec{q}\in \rvec{L}_{\Gin}\cup\lvec{R}_{\Gin}$\;

  \For{$\dvec{q}\in \rvec{L}_{\Gin}\cup\lvec{R}_{\Gin}$}{
    $h[\dvec{q}] := \FindPath(\dvec{q})$\;
  }

  \caption{Directed composition $h=f\diamond g$}
  \label{Alg:d-comp-refactored}
\end{algorithm}
\end{figure}

Throughout its execution each edge of~${\cal I}(f,g)$ 
is traversed at most once before the corresponding vertex is marked, 
so the total running time is $O(|E|+|V|)=O(|Q_\A|)$. 
Observe that $D(f,g,\dvec{q})$ is defined exactly as the endpoint of the 
path from~$\dvec{q}$ in the intermediate graph ${\cal I}(f,g)$. 
Hence a single call $D(f,g,\dvec{q})$ is realized by 
$\FindPath(\dvec{q})$ and runs in $O(|Q_\A|)$ time.
Therefore $\diamond$ and $D$ are computable in $O(|Q_\A|)$ time, 
and together with the bound for $CF$ this completes the proof. 
\qed\end{proof}

\subsection{Main Result}

We now summarize the simulation in the following theorem.

\begin{theorem}\label{thm:main}
For every $d(n)$-DLA $\A$, the membership problem for $\A$ is solvable 
in time $O(k \cdot n \cdot d(n) + m)$ on a RAM model, 
where $n$ is the length of the input word~$w$, 
$k=|Q_\A|$ is the number of states of $\A$, 
and $m$ is the length of the description of $\A$, 
assuming that the function $d(n)$ is computable in time $O(n+m)$.
\end{theorem}

\begin{proof}
By Lemma~\ref{Th::Correctness}, for each $d(n)$-DLA $\A$ 
there exists a corresponding DLLA $M$ 
(described by Algorithm~\ref{Fig::SimulationAlgorithm})
that simulates $\A$.      
By Lemma~\ref{th:AmortizedAnalysis}, $M$ performs 
$O(d(n)\cdot\langle U\!B\rangle(\A)\cdot n)$ steps, 
where $\langle U\!B\rangle(\A)$ is defined by Eq.~\eqref{eq:UB}.
By Lemma~\ref{lem:UB}, $\langle U\!B\rangle(\A) = O(|Q_\A|) = O(k)$. 
To simulate $M$ on a RAM model we use 
Simulation Algorithm~\ref{Fig::SimulationAlgorithm}
with subroutines from Algorithms~\ref{Fig:Alg:FindPath} 
and~\ref{Alg:d-comp-refactored}.
Before running the simulation we preprocess the description of $\A$ 
and store it in RAM, which takes $O(m)$ time.
Combining all these bounds, we obtain the claimed complexity 
$O(k \cdot n \cdot d(n) + m)$.
\end{proof}

\begin{corollary}
The recognition problem for a $d(n)$-DLA $\A$ is solvable in $O(n\cdot d(n))$ time.  
In particular, for each fixed $d \in \NN$, every $d$-DCFL is recognizable in linear time. 
\end{corollary}

\section{Upper bound on $d(n)$-DLA runtime}

In this section we establish an upper bound on the runtime of a $d(n)$-DLA in the classical simulation model (as defined above). 

\begin{theorem}\label{Thm:UpperBound}
A $d(n)$-DLA with $k$ states performs at most $O(d(n)\cdot n^2\cdot k)$ steps on an input of length $n$, 
provided the computation does not enter an infinite loop. 
\end{theorem}

\begin{proof}
Suppose the head traverses a segment of tape consisting of rank-$d(n)$ symbols 
for more than $kn$ steps. Then some cell must be visited at least twice in the same state, 
since the average number of state visits per cell exceeds $k$. Hence the computation would fall into an infinite loop.

As noted in the proof of Lemma~\ref{Th::Correctness}, $d(n)$-DLAs have two types of moves: 
\emph{regular moves}, when the head arrives at a cell of rank less than $d(n)$, 
and \emph{$d(n)$-moves}, when the head arrives at a segment consisting of rank-$d(n)$ cells. 
Any series of $d(n)$-moves cannot exceed $kn$ steps unless the computation enters a loop. 
Each such series must be preceded by a regular move, and the total number of regular moves is 
$O(n\cdot d(n))$: after each regular move the rank of the visited cell strictly increases, and it can increase only up to~$d(n)$. 

Since each regular move takes $O(1)$ steps, and each subsequent series of $d(n)$-moves 
takes $O(kn)$ steps, the total number of steps is bounded by $O(d(n)\cdot n^2\cdot k)$. 
Therefore the claim follows. \qed
\end{proof}

Theorem~\ref{Thm:UpperBound} implies that a $d$-DLA runs in $O(d\cdot k\cdot n^2)$ steps. 
This bound is asymptotically tight: for instance, the classical 
$2$-LA recognizing the language $\{a^n b^n \mid n \geq 0\}$ runs in quadratic time. 
Its behavior is as follows. The head moves right until it encounters the first $b$, 
then returns left to find the leftmost $a$ of rank~$1$ (which is then promoted to rank~$2$). 
After this $a$ is located, the automaton moves right again to check for a matching $b$ of rank~$1$; 
if found, it proceeds to the next matching $a$, and so on. 
When the automaton reaches the right endmarker~$\rendmarker$, it verifies that no $a$’s of rank~$1$ remain. 
In this case the input is accepted, and otherwise rejected.

\section*{Acknowledgments}
The author thanks Dmitry Chistikov for valuable feedback, 
for discussions of the results presented in this text, 
and for helpful suggestions that improved the presentation.

\bibliographystyle{splncs04}
%\bibliographystyle{plain}                                                                                                                                     
%\bibliography{flf_paper_springer}

\bibliography{flf_paper_lipics_short}

\end{document}